\documentclass{article}
\usepackage{amsfonts,amsmath,amssymb,latexsym}
\usepackage{eucal}
\usepackage[T1]{fontenc}
\usepackage[latin9]{inputenc}
\usepackage{amstext}
\usepackage{amsthm}
\usepackage{color}

\numberwithin{equation}{section}
\numberwithin{figure}{section}
\theoremstyle{plain}
\newtheorem{thm}{\protect\theoremname}
\theoremstyle{plain}
\newtheorem{prop}[thm]{\protect\propositionname}
\theoremstyle{definition}
\newtheorem{example}[thm]{\protect\examplename}
\theoremstyle{plain}
\newtheorem{lem}[thm]{\protect\lemmaname}

\allowdisplaybreaks

\providecommand{\examplename}{Example}
\providecommand{\lemmaname}{Lemma}
\providecommand{\propositionname}{Proposition}
\providecommand{\theoremname}{Theorem}

\begin{document}

%\begin{frontmatter}
%
%\journal{Physica A: Statistical Mechanics and its Applications}

\title{Conditions for the Existence of a Generalization of R\'{e}nyi Divergence}

\author{Rui F. Vigelis\thanks{Federal University of Cear\'{a}, Computer Engineering, Campus Sobral,  R. Cel. Estanislau Frota, Sobral - CE, 62010-560, Brazil, email: \texttt{rfvigelis@ufc.br}.}, Luiza H. F. de Andrade\thanks{Federal Rural University of the Semi-arid Region, Center of Exact and Natural Sciences, Av. Francisco Mota 572,  Mossor\'{o} - RN, 59625-900, Brazil, email:\texttt{luizafelix@ufersa.edu.br}.} and Charles C. Cavalcante\thanks{Federal University of Cear\'{a}, Teleinformatics Engineering Department, Campus do Pici, Bloco 722, CP 6005, 60440-900, Fortaleza - CE, Brazil, email:\texttt{charles@gtel.ufc.br}.}}

\date{}

%\author[ufc1]{}
%\ead{}
%
%\author[ufersa]{}
%\ead{luizafelix@ufersa.edu.br}
%
%\author[ufc2]{}
%\ead{charles@gtel.ufc.br}
%
%
%\address[ufc1]{Federal University of Cear\'{a}, Computer Engineering, Campus Sobral,  R. Cel. Estanislau Frota, Sobral - CE, 62010-560, Brazil}
%
%\address[ufersa]{Federal Rural University of the Semi-arid Region, Center of Exact and Natural Sciences, Av. Francisco Mota 572,  Mossor\'{o} - RN, 59625-900, Brazil}
%
%\address[ufc2]{Federal University of Cear\'{a}, Teleinformatics Engineering Department, Campus do Pici, Bloco 722, CP 6005, 60440-900, Fortaleza - CE, Brazil}

\maketitle

\begin{abstract}
We give necessary and sufficient conditions for the existence of a generalization of R\'{e}nyi divergence, which is defined in terms of a deformed exponential function. If the underlying measure $\mu$ is non-atomic, we found that not all deformed exponential functions can be used in the generalization of R\'{e}nyi divergence; a condition involving the deformed exponential function is provided. In the case $\mu$ is purely atomic (the counting measure on the set of natural numbers), we show that any deformed exponential function can be used in the generalization.
\end{abstract}

\begin{center}
  \begin{minipage}[c]{0.8\textwidth}
  \end{minipage}
\end{center}

%\begin{keyword}
%
%\end{keyword}
%
%\end{frontmatter}

%\maketitle

\section{Introduction}

Entropy has been widely employed as a key measure of information in dynamical systems. Information theory, the field that investigates the characterization and limits of information, allows a number of applications that span from areas such as communications, neurobiology, natural language processing, econometric and other physical systems \cite{Cover:2006}.

Shannon \cite{Shannon_1948} was the first to interpret that information was linked to probability and to propose the quantity as an information or uncertainty measure, which can be written as
\begin{equation*}
H(\mathbf{p})=-\sum_{i=1}^n p_{i}\ln p_{i},
\end{equation*}
where $\mathbf{p}$ is the  probability mass function of the source of information. The quantity was named as entropy by its similarity with Boltzmann entropy (see, for instance, \cite{Boltzmann_1964}). Another well know measure of information was proposed by Tsallis \cite{Tsallis:1988}, who defined the expression
\begin{equation*}
H_{q}(\mathbf{p})=\frac{1}{q-1}\left(1-\sum_{i=1}^{n}p_{i}^{q}\right),
\end{equation*}
as a generalized entropy dependent on the parameter $q\geq0$, since when $q\rightarrow1$ it reduces to the Shannon entropy. Also, in \cite{Tsallis:1994} Tsallis defined the function $\ln_{q}(x)=\frac{x^{1-q}-1}{1-q}$ for any non negative $q$, as a generalized logarithm function, which was termed as $q$-logarithm, since $\ln_{q}(x)\rightarrow\ln(x)$, as $q\rightarrow1$. As a consequence, Tsallis entropy generalizes Shannon entropy \cite{Tsallis:1994}. The uniqueness theorem for Tsallis entropy was presented in \cite{Suyari_2004} by introducing a generalization of Shannon-Khinchin axiom. Furthermore, this theorem was generalized and simplified in \cite{Furuichi_2005}. Tsallis entropy plays a crucial role in nonextensive statistics also called Tsallis statistics \cite{Tsallis_2001}.

On another way of visualizing the uncertainty of events and how to measure them, R\'{e}nyi proposed a family if entropies that can be written as \cite{Renyi:1961}:
\begin{equation*}
  H_{\alpha}(\mathbf{p}) = \frac{1}{1-\alpha}\ln\left(\sum_{i=1}^{n}p_{i}^{\alpha}\right),
\end{equation*}
where $\alpha$ is the entropy order. R\'{e}nyi entropy is then flexible in the sense it can, as in the Tsallis case, to provide several different expressions by choosing different entropy orders. Due to its properties for the case $\alpha = 2$, some researchers have been working in the field termed information theoretic learning (ITL) where several interesting properties arise for this entropy definition \cite{Principe:10}. This flexibility of the model proposed by R\'{e}nyi is one of our key interests on the investigation of work.

While entropy is an uncertainty measure, relative entropy can be interpreted as a measure of statistical distance between two probability distributions \cite{Cover:2006}. Relative entropy, or statistical divergence, plays an essential role in information geometry \cite{Zhang:2004}. A well-known example is the relative (Shannon) entropy, or Kullback-Leibler (KL) divergence, given by
\begin{equation*}
\mathcal{D}(\mathbf{p}\mathbin{||}\mathbf{q})=\sum_{i=1}^{n}p_{i}\ln\left(\frac{p_{i}}{q_{i}}\right),
\end{equation*}
which was defined in \cite{kullback1951}. It can also be interpreted as an analogous (non symmetric) of the squared
of the Euclidean distance \cite{Csiszar:2004}. One possible generalization of this divergence is the Tsallis relative entropy \cite{Tsallis_1998} which is obtained when we replace the ordinary logarithm by the $q$-logarithm
in the KL divergence which yields \begin{equation*}
\mathcal{D}_{q}(\boldsymbol{p}\mathbin{||}\boldsymbol{q})=\sum_{i=1}^{n}p_{i}\ln_{q}\Bigl(\frac{p_{i}}{q_{i}}\Bigr).
\end{equation*}
Both KL divergence and Tsallis relative entropy satisfy important properties, such as non negativity, monotonicity and joint convexity, among others \cite{Furuichi:2004}. One must note that the expressions for the definitions of entropy and divergence have considered discrete distributions, as in the original works  but it is straightforward to provide those expressions considering continuous distributions by replacing the summation by integrals and the probability mass function by the probability density function \cite{Cover:2006}.

The investigation of more general divergences and study their properties has been the object of interest of many researchers in the last decades. The interest on a different statistical divergence metric is motivated, among others, in applications related to optimization and statistical learning since more flexible functions and expressions may be suitable to larger classes of data and signals and lead to more efficient information recovery methods \cite{Hastie_2009,Principe_2010,Konishi_2008}. To cite a few, the usage of divergence metric has been considered in several domains such as statistics (including statistical physics) and learning \cite{Yamano:2009,Principe:10,Lionel:2013,Nielsen:2020}, econometrics \cite{nock2,trivellato,trivellato2,CP,Flavia:2019}, digital communications \cite{Alvarez:97,Santamaria:02,Cavalcante:02,Cavalcante:05_SP_Elsevier}, signal and image processing \cite{Yannick:2013,Fatima:2016,Rui:2019}, biomedical processing \cite{Nielsen:2011}. Also, quantum versions of generalized divergences are of interest in the literature \cite{Abe_2003,Luiza:2019}.

The general rationale on the consideration of divergence in such optimization problems is usually to derive more robust (or suitable) metrics to statistically differentiate two distributions stating how close (or how different) they are from each other. Csisz\'{a}r introduced yet another concept of divergence, the $f$-divergence defined as $\sum_{i}^{n}q_{i}f\left(\frac{p_{i}}{q_{i}}\right)$ \cite{Csiszar:2004} for any convex function $f(t)$ for $t>0$ such that $f(1)=0$. KL divergence and Tsallis relative entropy are also obtained as a particular case of the $f$-divergence. \textit{Amari $\alpha$-divergence} \cite{Amari1982,Amari1985,Amari:2000} is yet another divergence that can be seen as a special case of the $f$-divergence since such divergence reduces to the KL one, when $\alpha=\pm1$ . Bregman introduced in \cite{Bregman1967} a divergence which is induced by a convex differentiable function. In \cite{Zhang:2004} a more general expression of a divergence function was introduced, the $(\rho,\tau)$-divergence, that has as a especial case the \textit{Zhang's $\alpha$-divergence} which is based on the quasiarithmetic mean \cite{Inequalities_1952} and includes the Bregman divergence, the Amari $\alpha$-divergence and the $f$-divergence as special cases. Furthermore, more recently, the $(\rho,\tau)$-embedding was studied in \cite{Zhang_2015,Naudts_Zhang_2018} and Jain and Chhabra \cite{Chhabra_2016} introduced a new generalized divergence measure for increasing functions.

Some of the proposed generalized divergences rely on a more flexible function in order to exploit other statistical characteristics. The deformed exponentials proposed by Naudts \cite{Naudts_2002} and further investigated in the context of statistical physics in \cite{naudts5} are one of such more flexible models. The idea of those deformed functions is that they relax some conditions of the classical exponential function and expand the number of degrees of freedom one can play so aspects of heavier tails, for example, can be more easily incorporated within the same framework. This is particularly of interest in some problem in econophysics when the distribution of the risk changes due to some external aspects of the economy (unforeseen events such as pandemics and crash of stock market).

One generalization of the exponential families of probability distributions was introduced in \cite{Vigelis:2013}, with the so-called $\varphi$-families of probability distributions. This generalization was possible by replacing of the exponential function by a deformed exponential $\varphi$, with some appropriate conditions. In that work, the $\varphi$-divergence
was defined between two probability distributions in the same $\varphi$-family. The $\varphi$-divergence can be interpreted as the Bregman divergence associated to the normalizing function $\psi$, that is a convex differentiable function. Actually, the $\varphi$-divergence between two probability distributions $\mathbf{p}$ and $\mathbf{q}$, that are in
the same $\varphi$-family, is the normalizing function $\psi$, that appears when we write the probability distribution $\mathbf{q}$ as a function of $\mathbf{p}$. In others words, $\varphi$-divergence appears naturally in the theory of information geometry \cite{Korbel_2019}. Furthermore, the $\varphi$-divergence has an inherent relationship
with Zhang's $\alpha$-divergence \cite{Zhang:2004}.

%(3) Is there any example where failure to meet the existence conditions studied here leads to problems in some concrete
%application of generalized divergences in physics or other areas?
%(4) The authors should include a conclusions section, where they provide a summary of the main results obtained,
%and of their relevance for physics.

%%% Old text of introduction

R\'{e}nyi divergence \cite{Renyi:1961} is one of the most successful
measures of dissimilarity between probability distributions, having
found many applications \cite{Erven:2014}. It is given by
\begin{equation*}
\mathcal{D}^{(\alpha) }(\mathbf{p}||\mathbf{q})={\frac {1}{\alpha -1}}\ln {\Bigg (}\sum _{i=1}^{n}{\frac {p_{i}^{\alpha }}{q_{i}^{\alpha -1}}}{\Bigg )},
\end{equation*}
where $\alpha$ is the order of the entropy (a free parameter). In \cite{Souza:2016}
the authors proposed a generalization of R\'{e}nyi divergence in terms
of a deformed exponential function. In order that this generalization
be well-defined, the deformed exponential function have to satisfy
some suitable conditions, which we investigate in the present paper.
We considered the cases in which the underlying measure $\mu$ is
non-atomic or purely atomic (the counting measure on the set of natural
numbers $\mathbb{N}$). Each case required distinctive techniques,
and provided different results. If the measure $\mu$ is non-atomic,
we found that not all deformed exponential functions can be used in
the generalization of R\'{e}nyi divergence; a condition involving the
deformed exponential function is provided. In the case $\mu$ is the
counting measure on $\mathbb{N}$, we prove that any deformed exponential
function can be used to define the generalization of R\'{e}nyi divergence.
These results are found in Section~\ref{sec:existence_conditions}.
In what follows, we show how the R\'{e}nyi divergence can be generalized
in terms of a deformed exponential function; the limit cases are also
discussed.

Let $(T,\Sigma,\mu)$ be a $\sigma$-finite measure space. All probability
distributions (or probability measures) are assumed to have positive
density w.r.t.\ the underlying measure $\mu$. In other words, they
belong to the collection
\[
\mathcal{P}_{\mu}=\biggl\{ p\in L^{0}:\int_{T}pd\mu=1\text{ and }p>0\biggr\},
\]
where $L^{0}$ is the space of all real-valued, measurable functions
on $T$, with equality $\mu$-a.e.

The R\'{e}nyi divergence of order $\alpha\in(0,1)$ between probability
distributions $p$ and $q$ in $\mathcal{P}_{\mu}$ is defined as
\begin{equation}
\mathcal{D}^{(\alpha)}(p\parallel q)=\frac{\kappa(\alpha)}{\alpha(\alpha-1)},\label{eq:Renyi_divergence}
\end{equation}
where
\begin{equation}
\kappa(\alpha)=-\log\biggl(\int_{T}p^{\alpha}q^{1-\alpha}d\mu\biggr).
\label{Eq:Kappa_alpha}
\end{equation}
For $\alpha\in\{0,1\}$, the R\'{e}nyi divergence is defined by taking
a limit:
\begin{align*}
\mathcal{D}^{(0)}(p\parallel q) & =\lim_{\alpha\downarrow0}\mathcal{D}^{(\alpha)}(p\parallel q),\\
\mathcal{D}^{(1)}(p\parallel q) & =\lim_{\alpha\uparrow1}\mathcal{D}^{(\alpha)}(p\parallel q).
\end{align*}
Expression \eqref{eq:Renyi_divergence} can be used to define the
R\'{e}nyi divergence $\mathcal{D}^{(\alpha)}(\cdot\parallel\cdot)$ for
every $\alpha\in\mathbb{R}$. However, for $\alpha\notin(-1,1)$ this
expression may not be finite-valued for all $p$ and $q$ in $\mathcal{P}_{\mu}$.
To avoid some technicalities, we assume that $\alpha\in[-1,1]$. The
standard form of the R\'{e}nyi divergence found in the literature differs
from \eqref{eq:Renyi_divergence} by a factor of $1/\alpha$. We chose
to define $\mathcal{D}^{(\alpha)}(\cdot\parallel\cdot)$ as in \eqref{eq:Renyi_divergence}
so that some symmetry could be preserved when the limits $\alpha\downarrow0$
and $\alpha\uparrow1$ are taken.

The generalization of R\'{e}nyi divergence is based on an alternate interpretation
of $\kappa(\alpha)$. Fixed $\alpha\in(0,1)$, and given any $p$
and $q$ in $\mathcal{P}_{\mu}$, the function $\kappa(\alpha):=\kappa(\alpha;p,q)$
is the unique non-negative real number such that
\[
\int_{T}\exp(\alpha\ln(p)+(1-\alpha)\ln(q)+\kappa(\alpha))d\mu=1.
\]
To generalize the R\'{e}nyi divergence, we consider a deformed exponential
$\varphi(\cdot)$ in the place of the exponential function. A \textit{deformed
exponential} $\varphi\colon\mathbb{R}\rightarrow[0,\infty)$ is a
convex function such that $\lim_{u\rightarrow-\infty}\varphi(u)=0$
and $\lim_{u\rightarrow\infty}\varphi(u)=\infty$. Given any $p$
and $q$ in $\mathcal{P}_{\mu}$, we take $\kappa(\alpha)=\kappa(\alpha;p,q)\geq0$
so that
\begin{equation}
\int_{T}\varphi(\alpha\varphi^{-1}(p)+(1-\alpha)\varphi^{-1}(q)+\kappa(\alpha)u_{0})d\mu=1,\label{eq:kappa_phi}
\end{equation}
where $u_{0}\colon T\rightarrow(0,\infty)$ is a positive, measurable
function satisfying a suitable condition. The existence and uniqueness
of $\kappa(\alpha)$ as defined in \eqref{Eq:Kappa_alpha} is guaranteed
by the condition in \eqref{eq:kappa_phi}, which will be investigated in the next section.
We will show that the existence of $u_{0}$ depends on $\varphi(\cdot)$
and the underlying measure $\mu$.

We define the generalization of R\'{e}nyi divergence of order $\alpha\in(0,1)$
by
\begin{equation}
\mathcal{D}_{\varphi}^{(\alpha)}(p\parallel q)=\frac{\kappa(\alpha)}{\alpha(1-\alpha)},\label{eq:D_phi}
\end{equation}
where $\kappa(\alpha)$ is given as in \eqref{eq:kappa_phi}. For
$\alpha\in\{0,1\}$, the generalization is defined by taking a limit:
\begin{align}
\mathcal{D}_{\varphi}^{(0)}(p\parallel q) & =\lim_{\alpha\downarrow0}\mathcal{D}_{\varphi}^{(\alpha)}(p\parallel q),\label{eq:D_phi_zero}\\
\mathcal{D}_{\varphi}^{(1)}(p\parallel q) & =\lim_{\alpha\uparrow1}\mathcal{D}_{\varphi}^{(\alpha)}(p\parallel q).\label{eq:D_phi_one}
\end{align}
These limits are related to a generalization of the Kullback\textendash Leibler
divergence \cite{Kullback:1951}, the so-called $\varphi$-divergence,
which was introduced by the authors in \cite{Vigelis:2013}. The $\varphi$-divergence
is given by
\begin{equation}
\mathcal{D}_{\varphi}(p\parallel q)=\frac{\int_{T}\frac{\varphi^{-1}(p)-\varphi^{-1}(q)}{(\varphi^{-1})'(p)}d\mu}{\int_{T}\frac{u_{0}}{(\varphi^{-1})'(p)}d\mu}.\label{eq:phi_divergence}
\end{equation}
In the case $\varphi(\cdot)$ is the exponential function and $u_{0}=1$,
the $\varphi$-divergence reduces to the Kullback\textendash Leibler
divergence. Under some conditions, the limits \eqref{eq:D_phi_zero}
and \eqref{eq:D_phi_one} are finite-valued and converges to the $\varphi$-divergence:
\begin{equation}
\mathcal{D}_{\varphi}^{(0)}(q\parallel p)=\mathcal{D}_{\varphi}^{(1)}(p\parallel q)=\mathcal{D}_{\varphi}(p\parallel q)<\infty.\label{eq:D_zero_D_one_relation}
\end{equation}
These conditions are stated in Proposition~\ref{prop:D_phi_limit}
for the case involving the generalized R\'{e}nyi divergence.

\begin{prop}
\label{prop:D_phi_limit} Assume that $\varphi(\cdot)$ is continuously
differentiable. Consider the condition
\begin{equation}
\int_{T}\varphi(\alpha\varphi^{-1}(p)+(1-\alpha)\varphi^{-1}(q))d\mu<\infty.\label{eq:cond_int}
\end{equation}
If expression \eqref{eq:cond_int} is satisfied for all $\alpha\in[\alpha_{0},0)$
and some $\alpha_{0}<0$, then
\[
\mathcal{D}_{\varphi}^{(0)}(p\parallel q)=\frac{\partial\kappa}{\partial\alpha}(0)=\mathcal{D}_{\varphi}(q\parallel p)<\infty.
\]
If expression \eqref{eq:cond_int} is satisfied for all $\alpha\in(1,\alpha_{0}]$
and some $\alpha_{0}>1$, then
\[
\mathcal{D}_{\varphi}^{(1)}(p\parallel q)=-\frac{\partial\kappa}{\partial\alpha}(1)=\mathcal{D}_{\varphi}(p\parallel q)<\infty.
\]
\end{prop}

Notice that expression \eqref{eq:cond_int} always holds for $\alpha\in[0,1]$,
since $\varphi(\cdot)$ is convex. For a proof of Proposition~\ref{prop:D_phi_limit}
we refer to Lemma~4 and Proposition~5 in \cite{Souza:2016}.

\section{Existence Conditions\label{sec:existence_conditions}}

The generalization of R\'{e}nyi divergence requires that $\kappa(\alpha)$
be well-defined. To guarantee the existence and uniqueness of $\kappa(\alpha)$
as defined by \eqref{eq:kappa_phi}, we assume that there exists a
measurable function $u_{0}\colon T\rightarrow(0,\infty)$ such that
\begin{equation}
\int_{T}\varphi(c+\lambda u_{0})d\mu<\infty,\qquad\text{for all }\lambda>0,\label{eq:u0_condition}
\end{equation}
for each measurable function $c\colon T\rightarrow\mathbb{R}$ satisfying
$\int_{T}\varphi(c)d\mu<\infty$. The existence of $u_{0}$ depends
on the deformed exponential $\varphi(\cdot)$ and the underlying measure
$\mu$. In the case $\mu$ is non-atomic, not all deformed exponential
functions admit the existence of a function $u_{0}$ satisfying \eqref{eq:u0_condition}.
(A measure $\mu$ is said to be \textit{non-atomic} if for any measurable
set $A$ with $\mu(A)>0$ there exists a measurable subset $B\subset A$
such that $\mu(A)>\mu(B)>0$.) We shall find a condition involving
solely $\varphi(\cdot)$ which is equivalent to the existence of $u_{0}$.
If $\mu$ is the counting measure on the set of natural numbers $T=\mathbb{N}$,
we will show that, for any deformed exponential function $\varphi(\cdot)$,
always there exists a function $u_{0}$ (to be more precise, a sequence)
satisfying \eqref{eq:u0_condition}.

Many deformed exponential functions $\varphi(\cdot)$ can be used
in the generalization of R\'{e}nyi divergence. A standard example is the
exponential function, which satisfies condition \eqref{eq:u0_condition}
for $u_{0}=1$. Another example is the Kaniadakis' $\kappa$-exponential
\cite{Kaniadakis:2002,Vigelis:2013}. For the deformed exponential
function given below, we cannot find a function $u_{0}$ for which
condition \eqref{eq:u0_condition} holds.

\begin{example}
\label{exa:u0_not_satisfied} Let us consider the deformed exponential
function
\[
\varphi(u)=\begin{cases}
e^{(u+1)^{2}/2}, & u\geq0,\\
e^{(u+\frac{1}{2})}, & u\leq0.
\end{cases}
\]
Assume that the underlying measure $\mu$ is non-atomic. Given any
measurable function $u_{0}\colon T\rightarrow(0,\infty)$, we can
find a measurable function $c\colon T\rightarrow\mathbb{R}$ with
$\int_{T}\varphi(c)d\mu<\infty$, for which condition \eqref{eq:u0_condition}
does not hold. This claim was proved by the authors in \cite[Example 2]{Souza:2016}.
An alternate proof of this result follows from a proposition (which
will be shown in this section) involving the existence of $u_{0}$.
\end{example}

The next result shows that condition \eqref{eq:u0_condition} is appropriate
for the existence of $\kappa(\alpha)$, since they are equivalent.

\begin{prop}
Assume that the measure $\mu$ is non-atomic. Fix any $\alpha\in(0,1)$.
A deformed exponential $\varphi\colon\mathbb{R}\rightarrow[0,\infty)$
and a measurable function $u_{0}\colon T\rightarrow(0,\infty)$ satisfy
 condition \eqref{eq:u0_condition} if, and only if, for each probability
distributions $p$ and $q$ in $\mathcal{P}_{\mu}$, there exists
a constant $\kappa(\alpha):=\kappa(\alpha;p,q)$ such that
\begin{equation}
\int_{T}\varphi(\alpha\varphi^{-1}(p)+(1-\alpha)\varphi^{-1}(q)+\kappa(\alpha)u_{0})d\mu=1.\label{eq:kappa_phi_prop}
\end{equation}
\end{prop}

\begin{proof}
If condition \eqref{eq:u0_condition} is satisfied, the existence
and uniqueness of $\kappa(\alpha)$ follows from the Monotone Convergence
Theorem and the continuity of $\varphi(\cdot)$.

Suppose that condition \eqref{eq:u0_condition} does not hold. In
this case, for some measurable function $c\colon T\rightarrow\mathbb{R}$
with $\int_{T}\varphi(c)d\mu<\infty$, and some $\lambda_{0}\geq0$,
we have
\begin{equation}
\left\{ \begin{alignedat}{2}\int_{T}\varphi(c+\lambda u_{0})d\mu & <\infty, &  & \qquad\text{for }0\leq\lambda\leq\lambda_{0},\\
\int_{T}\varphi(c+\lambda u_{0})d\mu & =\infty, &  & \qquad\text{for }\lambda_{0}<\lambda,
\end{alignedat}
\right.\label{eq:u0_not_sat_leq}
\end{equation}
or
\begin{equation}
\left\{ \begin{alignedat}{2}\int_{T}\varphi(c+\lambda u_{0})d\mu & <\infty, &  & \qquad\text{for }0\leq\lambda<\lambda_{0},\\
\int_{T}\varphi(c+\lambda u_{0})d\mu & =\infty, &  & \qquad\text{for }\lambda_{0}\leq\lambda.
\end{alignedat}
\right.\label{eq:u0_not_sat_less}
\end{equation}
Notice that \eqref{eq:u0_not_sat_less} cannot be satisfied for $\lambda_{0}=0$.
Let $\{T_{n}\}$ be a sequence of non-decreasing, measurable sets
with $\mu(T_{n})<\infty$ and $\mu(T\setminus\bigcup_{n=1}^{\infty}T_{n})=0$.
Define $A_{n}=T_{n}\cap\{c\leq n\}\cap\{u_{0}\leq n\}$, for each
$n\geq1$. Clearly, the sequence $\{A_{n}\}$ is non-decreasing and
satisfies $\mu(A_{n})<\infty$ and $\mu(T\setminus\bigcup_{n=1}^{\infty}A_{n})=0$.
Moreover,
\[
\int_{A_{n}}\varphi(c+\lambda u_{0})d\mu\leq\varphi(n+\lambda n)\mu(A_{n})<\infty,
\]
for all $\lambda>0$, and each $n\geq1$.

If the function $u_{0}$ satisfies \eqref{eq:u0_not_sat_leq}, we
select a sufficiently large $n_{0}\geq1$ such that $\int_{T\setminus A_{n_{0}}}\varphi(c+\lambda_{0}u_{0})d\mu<1$.
Denote $B:=T\setminus A_{n_{0}}$. Let $b_{1},b_{2}\colon T\rightarrow\mathbb{R}$
be measurable functions for which $p=\varphi(c_{1})$ and $q=\varphi(c_{2})$
are in $\mathcal{P}_{\mu}$, where $c_{1}=b_{1}\chi_{T\setminus B}+(c+\lambda_{0}u_{0})\chi_{B}$
and $c_{2}=b_{2}\chi_{T\setminus B}+(c+\lambda_{0}u_{0})\chi_{B}$.
Moreover, we assume $b_{1}\chi_{T\setminus B}\neq b_{2}\chi_{T\setminus B}$.
For any $\lambda>0$, we can write
\begin{multline*}
\int_{T}\varphi(\alpha\varphi^{-1}(p)+(1-\alpha)\varphi^{-1}(q)+\lambda u_{0})\geq\int_{B}\varphi(c+(\lambda_{0}+\lambda)u_{0})d\mu\\
=\int_{T}\varphi(c+(\lambda_{0}+\lambda)u_{0})d\mu-\int_{A_{n_{0}}}\varphi(c+(\lambda_{0}+\lambda)u_{0})d\mu=\infty.
\end{multline*}
Thus, the constant $\kappa(\alpha)$, as defined by \eqref{eq:kappa_phi_prop},
cannot be found.

Now suppose that \eqref{eq:u0_not_sat_less} is satisfied. Let $\{\lambda_{n}\}$
be a sequence in $(0,\lambda_{0})$ such that $\lambda_{n}\uparrow\lambda_{0}$.
We define inductively an increasing sequence $\{k_{n}\}\subseteq\mathbb{N}$
as follows. Choose $k_{1}\geq1$ such that $\int_{A_{k_{1}}}\varphi(c+\lambda_{0}u_{0})d\mu\geq1$
and $\int_{A_{k_{1}}}\varphi(c+\lambda_{1}u_{0})d\mu\leq2^{-(1+1)}$.
Given $k_{n-1}$ we select some $k_{n}>k_{n-1}$ such that
\[
\int_{A_{k_{n}}\setminus A_{k_{n-1}}}\varphi(c+\lambda_{0}u_{0})d\mu\geq1
\]
and
\[
\int_{A_{k_{n}}\setminus A_{k_{n-1}}}\varphi(c+\lambda_{n}u_{0})d\mu\leq2^{-(n+1)}.
\]
Let us denote $B_{1}=A_{k_{1}}$ and $B_{n}=A_{k_{n}}\setminus A_{k_{n-1}}$
for $n>1$. Notice that the sets $B_{n}$ are pairwise disjoint. Define
$u=\sum_{n=1}^{\infty}\lambda_{n}u_{0}\chi_{B_{n}}$ and $B=\bigcup_{n=1}^{\infty}B_{n}$.
As a result of this construction, it follows that
\[
\int_{B}\varphi(c+u)d\mu\leq\frac{1}{2}.
\]
Let $b_{1},b_{2}\colon T\rightarrow\mathbb{R}$ be measurable functions
for which $p=\varphi(c_{1})$ and $q=\varphi(c_{2})$ are in $\mathcal{P}_{\mu}$,
where $c_{1}=b_{1}\chi_{T\setminus B}+(c+u)\chi_{B}$ and $c_{2}=b_{2}\chi_{T\setminus B}+(c+u)\chi_{B}$.
In addition, we assume $b_{1}\chi_{T\setminus B}\neq b_{2}\chi_{T\setminus B}$.
Fixed arbitrary $\lambda>0$, we take $n_{1}\geq1$ such that $\lambda_{n}+\lambda\geq\lambda_{0}$
for all $n\geq n_{1}$. Observing that $\int_{B_{n}}\varphi(c+\lambda_{0}u_{0})d\mu\geq1$,
we can write
\begin{multline*}
\int_{T}\varphi(\alpha\varphi^{-1}(p)+(1-\alpha)\varphi^{-1}(q)+\lambda u_{0})d\mu\geq\int_{B}\varphi(c+u+\lambda u_{0})d\mu\\
\geq\sum_{n=n_{1}}^{\infty}\int_{B_{n}}\varphi(c+(\lambda_{n}+\lambda)u_{0})d\mu\geq\sum_{n=n_{1}}^{\infty}1=\infty,
\end{multline*}
which shows that $\kappa(\alpha)$ cannot be found.
\end{proof}

The analysis concerning the existence of $u_{0}$ implicates the use
of different techniques, which depend on the measure $\mu$ be non-atomic
or purely atomic (the counting measure on the set of natural numbers
$T=\mathbb{N}$).

\subsection{Non-atomic case}

As shown in Example~\ref{exa:u0_not_satisfied}, where the measure
$\mu$ was assumed to be non-atomic, not all deformed exponential
functions accept the existence of a function $u_{0}$ satisfying \eqref{eq:u0_condition}.
Supposing that $\mu$ is non-atomic, we will present an equivalent
criterion for a deformed exponential function and a function $u_{0}$
to satisfy condition \eqref{eq:u0_condition}. Using this result,
we will find a condition involving solely $\varphi(\cdot)$ which
is equivalent to the existence of $u_{0}$. Throughout this subsection,
we assume that the measure $\mu$ is non-atomic.

\begin{prop}
\label{prop:condition_ineq} A deformed exponential $\varphi\colon\mathbb{R}\rightarrow[0,\infty)$
and a measurable function $u_{0}\colon T\rightarrow(0,\infty)$ satisfy
condition \eqref{eq:u0_condition} if, and only if, for some constant
$\alpha\in(0,1)$, we can find a measurable function $c\colon T\rightarrow\mathbb{R}\cup\{-\infty\}$
such that $\int_{T}\varphi(c)d\mu<\infty$ and
\begin{equation}
\alpha\varphi(u)\leq\varphi(u-u_{0}(t)),\qquad\text{for all }u\geq c(t),\label{eq:ineq_condition_lemma_2}
\end{equation}
for $\mu$-a.e.\ $t\in T$.
\end{prop}

Inequalities similar to \eqref{eq:ineq_condition_lemma_2} will be
assumed to hold for $\mu$-a.e.\ $t\in T$. Accordingly, we will
omit this assumption hereafter. The proof of Proposition~\ref{prop:condition_ineq}
requires some preliminary results.

\begin{lem}
\label{lem:seq} Let $\mu$ be a non-atomic, $\sigma$-finite measure.
If $\{\alpha_{m}\}$ is a sequence of positive, real numbers, and
$\{u_{m}\}$ is a sequence of finite-valued, non-negative, measurable
functions, such that
\[
\int_{T}u_{m}d\mu\geq2^{m}\alpha_{m},\quad\text{for all }m\geq1,
\]
then there exist an increasing sequence $\{m_{n}\}$ of natural numbers
and a sequence $\{A_{n}\}$ of pairwise disjoint, measurable sets
such that
\[
\int_{A_{n}}u_{m_{n}}d\mu=\alpha_{m_{n}},\quad\text{for all }n\geq1.
\]
\end{lem}

A proof of Lemma~\ref{lem:seq} is found in \cite[Lemma 8.3]{Musielak:1983}.
We use Lemma~\ref{lem:seq} to prove the result stated below.

\begin{lem}
\label{lem:technical} Suppose that we cannot find $\alpha\in(0,1)$
and a measurable function $c\colon T\rightarrow\mathbb{R}\cup\{-\infty\}$
such that $\int_{T}\varphi(c)d\mu<\infty$ and
\begin{equation}
\alpha\varphi(u)\leq\varphi(u-u_{0}(t)),\qquad\text{for all }u\geq c(t).\label{eq:ineq_condition_lemma}
\end{equation}
Then there exist sequences $\{c_{n}\}$ and $\{A_{n}\}$ of measurable
functions, and pairwise disjoint, measurable sets, respectively, such
that
\begin{equation}
\int_{A_{n}}\varphi(c_{n})d\mu=1\quad\text{and}\quad\int_{A_{n}}\varphi(c_{n}-u_{0})d\mu\leq2^{-n},\quad\text{for all }n\geq1.\label{eq:c_n_lambda_n}
\end{equation}
\end{lem}

\begin{proof}
For each $m\geq1$, we define the function
\[
f_{m}(t)=\sup\{u\in\mathbb{R}:2^{-m}\varphi(u)>\varphi(u-u_{0}(t))\},
\]
where we use the convention $\sup\emptyset=-\infty$. We will verify
that $f_{m}$ is measurable. For each rational number $r$, define
the measurable sets
\[
E_{m,r}=\{t\in T:2^{-m}\varphi(r)>\varphi(r-u_{0}(t))\}
\]
and the simple functions $u_{m,r}=r\chi_{E_{m,r}}$. Let $\{r_{i}\}$
be an enumeration of the rational numbers. For each $m,k\geq1$, consider
the non-negative, simple functions $v_{m,k}=\max_{1\leq i\leq k}u_{m,r_{i}}$.
Moreover, denote $B_{m,k}=\bigcup_{i=1}^{k}E_{m,r_{i}}$. By the continuity
of $\varphi(\cdot)$, it follows that $\varphi(v_{m,k})\chi_{B_{m,k}}\uparrow\varphi(f_{m})$
as $k\rightarrow\infty$, which shows that $f_{m}$ is measurable.
Since \eqref{eq:ineq_condition_lemma} is not satisfied, we have that
$\int_{T}\varphi(f_{m})d\mu=\infty$ for all $m\geq1$. In virtue
of the Monotone Convergence Theorem, for each $m\geq1$, we can find
some $k_{m}\geq1$ such that the function $v_{m}=v_{m,k_{m}}$ and
the set $B_{m}=B_{m,k_{m}}$ satisfy $\int_{B_{m}}\varphi(v_{m})d\mu\geq2^{m}$.
Clearly, we have that $\varphi(v_{m})\chi_{B_{m}}<\infty$ and $2^{-m}\varphi(v_{m})\chi_{B_{m}}\geq\varphi(v_{m}-u_{0})\chi_{B_{m}}$.
By Lemma~\ref{lem:seq}, there exist an increasing sequence $\{m_{n}\}$
of indices and a sequence $\{A_{n}\}$ of pairwise disjoint, measurable
sets such that $\int_{A_{n}}\varphi(v_{m_{n}})d\mu=1$. Clearly, $\int_{A_{n}}\varphi(v_{m_{n}}-u_{0})d\mu\leq2^{-m_{n}}$.
Denoting $c_{n}=v_{m_{n}}$, we obtain \eqref{eq:c_n_lambda_n}.
\end{proof}

\begin{proof}[Proof of Proposition~\ref{prop:condition_ineq}]
 Assume that $\varphi(\cdot)$ and $u_{0}$ satisfy condition \eqref{eq:u0_condition}.
Suppose that expression \eqref{eq:ineq_condition_lemma_2} does not
hold. Let $\{c_{n}\}$ and $\{A_{n}\}$ be as stated in Lemma~\ref{lem:technical}.
Denote $A=\bigcup_{n=1}^{\infty}A_{n}$. Then we define $c=c_{0}\chi_{T\setminus A}+\sum_{n=1}^{\infty}c_{n}\chi_{A_{n}}$,
where $c_{0}\colon T\rightarrow\mathbb{R}$ is any measurable function
such that $\int_{T\setminus A}\varphi(c_{0})d\mu<\infty$. Using \eqref{eq:c_n_lambda_n},
we can write
\begin{align}
\int_{T}\varphi(c)d\mu & =\int_{T\setminus A}\varphi(c_{0})d\mu+\sum_{n=1}^{\infty}\int_{A_{n}}\varphi(c_{n})d\mu\nonumber \\
 & =\int_{T\setminus A}\varphi(c_{0})d\mu+\sum_{n=1}^{\infty}1=\infty.\label{eq:int_c_infty}
\end{align}
In addition, it follows that
\begin{align*}
\int_{T}\varphi(c-u_{0})d\mu & =\int_{T\setminus A}\varphi(c_{0}-u_{0})d\mu+\sum_{n=1}^{\infty}\int_{A_{n}}\varphi(c_{n}-u_{0})d\mu\\
 & \leq\int_{T\setminus A}\varphi(c_{0})d\mu+\sum_{n=1}^{\infty}2^{-n}<\infty.
\end{align*}
By condition \eqref{eq:u0_condition}, we get $\int_{T}\varphi(c)d\mu=\int_{T}\varphi(c-u_{0}+u_{0})d\mu<\infty$,
which is a contradiction to \eqref{eq:int_c_infty}.

Conversely, suppose that expression \eqref{eq:ineq_condition_lemma_2}
holds. Let $\widetilde{c}\colon T\rightarrow\mathbb{R}$ be any measurable
function satisfying $\int_{T}\varphi(\widetilde{c})d\mu<\infty$.
Denote $A=\{t:\widetilde{c}(t)+u_{0}(t)\geq c(t)\}$. We use inequality
\eqref{eq:ineq_condition_lemma_2} to write
\begin{align*}
\alpha\int_{T}\varphi(\widetilde{c}+u_{0})d\mu & \leq\alpha\int_{A}\varphi(\widetilde{c}+u_{0})d\mu+\alpha\int_{T\setminus A}\varphi(c)d\mu\\
 & \leq\int_{A}\varphi(\widetilde{c})d\mu+\alpha\int_{T\setminus A}\varphi(c)d\mu<\infty.
\end{align*}
As a result, we can conclude that $\int_{T}\varphi(\widetilde{c}+nu_{0})d\mu<\infty$
for all $n\geq1$. Consequently, $\int_{T}\varphi(\widetilde{c}+\lambda u_{0})d\mu<\infty$
for all $\lambda>0$.
\end{proof}

In Proposition~\ref{prop:condition_ineq}, if we consider $u_{0}=1$
then the function $c(t)$ can be chosen to be constant. Clearly inequality
\eqref{eq:ineq_condition_lemma_2} with $u_{0}=1$ holds for all $u\geq\operatorname{ess\,inf}c(t)$.
As a result, we can replace $c(t)$ with $\operatorname{ess\,inf}c(t)$
if the measure $\mu$ is finite. On the other hand, assume $\mu(T)=\infty$.
Then $\int_{T}\varphi(c)d\mu<\infty$ implies $\operatorname{ess\,inf}c(t)=a_{\varphi}:=\inf\{u\in\mathbb{R}:\varphi(u)>0\}$.
It cannot be the case $a_{\varphi}>-\infty$, since we would have
$0<\alpha\varphi(u)\leq\varphi(u-1)=0$ for $a_{\varphi}<u\leq a_{\varphi}+1$.
Consequently, the function $c(t)$ can be replaced by $\operatorname{ess\,inf}c(t)=-\infty$;
and inequality \eqref{eq:ineq_condition_lemma_2} holds for all $u\in\mathbb{R}$.

Next we present a class of deformed exponential functions which admit
$u_{0}=1$.

\begin{example}
We will show that the Kaniadakis' $\kappa$-exponential $\exp_{\kappa}(\cdot)$
and $u_{0}=1$ satisfy condition \eqref{eq:u0_condition}. The \textit{$\kappa$-exponential}\emph{
}$\exp_{\kappa}\colon\mathbb{R}\rightarrow(0,\infty)$ for $\kappa\in[-1,1]$
is defined as
\[
\exp_{\kappa}(u)=\begin{cases}
(\kappa u+\sqrt{1+\kappa^{2}u^{2}})^{1/\kappa}, & \text{if }\kappa\neq0,\\
\exp(u), & \text{if }\kappa=0,
\end{cases}
\]
Its inverse, the so called\textit{ $\kappa$-logarithm} $\log_{k}\colon(0,\infty)\rightarrow\mathbb{R}$,
is given by
\[
\log_{\kappa}(v)=\begin{cases}
\dfrac{v^{\kappa}-v^{-\kappa}}{2\kappa}, & \text{if }\kappa\neq0,\\
\ln(v), & \text{if }\kappa=0.
\end{cases}
\]
We will verify that there exist $\alpha\in(0,1)$ and $\lambda>0$
for which
\begin{equation}
\lambda\leq\log_{\kappa}(v)-\log_{\kappa}(\alpha v),\qquad\text{for all }v>0.\label{eq:ineq_log_kappa}
\end{equation}
Some manipulations imply that the derivative of $\log_{\kappa}(v)-\log_{\kappa}(\alpha v)$
is negative for $0<v\leq v_{0}$ and positive for $v\geq v_{0}$,
where
\[
v_{0}=\Bigl(\frac{\alpha^{-\kappa}-1}{1-\alpha^{\kappa}}\Bigr)^{\frac{1}{2\kappa}} = \left(\frac{1}{\alpha}\right)^{\frac{1}{2}} > 0.
\]
Consequently, the difference $\log_{\kappa}(v)-\log_{\kappa}(\alpha v)$
attains a minimum at $v_{0}$; given $\alpha\in(0,1)$, inequality
\eqref{eq:ineq_log_kappa} is satisfied for some $\lambda>0$. Inserting
$v=\exp_{\kappa}(u)$ into \eqref{eq:ineq_log_kappa}, we can write
\begin{equation}
\alpha\exp_{\kappa}(u)\leq\exp_{\kappa}(u-\lambda),\qquad\text{for all }u\in\mathbb{R}.\label{eq:ineq_exp_kappa_lambda}
\end{equation}
If $n\in\mathbb{N}$ is such that $n\lambda\geq1$, then a repeated
application of \eqref{eq:ineq_exp_kappa_lambda} yields
\[
\alpha^{n}\exp_{\kappa}(u)\leq\exp_{\kappa}(u-n\lambda)\leq\exp_{\kappa}(u-1),\qquad\text{for all }u\in\mathbb{R}.
\]
Proposition~\ref{prop:condition_ineq} implies that $u_{0}=1$ satisfies
condition \eqref{eq:u0_condition}.
\end{example}

Now we show an equivalent criterion for the existence of $u_{0}$
satisfying \eqref{eq:u0_condition}.

\begin{prop}
Let $\varphi\colon\mathbb{R}\rightarrow[0,\infty)$ be a deformed
exponential. Then we can find a measurable function $u_{0}\colon\mathbb{R}\rightarrow(0,\infty)$
for which condition \eqref{eq:u0_condition} holds if, and only if,
\begin{equation}
\limsup_{u\rightarrow\infty}\frac{\varphi(u)}{\varphi(u-\lambda_{0})}<\infty,\label{eq:phi_cond_equiv}
\end{equation}
for some $\lambda_{0}>0$.
\end{prop}

\begin{proof}
By Proposition~\ref{prop:condition_ineq} we can conclude that the
existence of $u_{0}$ implies \eqref{eq:phi_cond_equiv}. Conversely,
assume that expression \eqref{eq:phi_cond_equiv} holds for some $\lambda_{0}>0$.
In this case, there exists $M\in(1,\infty)$ and $\overline{c}\in\mathbb{R}$
such that $\frac{\varphi(u)}{\varphi(u-\lambda_{0})}\leq M$ for all
$u\geq\overline{c}$. Let $\{\lambda_{n}\}$ be any sequence in $(0,\lambda_{0}]$
such that $\lambda_{n}\downarrow0$. For each $n\geq1$, define
\begin{equation}
c_{n}=\sup\{u\in\mathbb{R}:\alpha\varphi(u)>\varphi(u-\lambda_{n})\},\label{eq:cn_def}
\end{equation}
where $\alpha=1/M$ and we adopt the convention $\sup\emptyset=-\infty$.
From the choice of $\{\lambda_{n}\}$ and $\alpha$, it follows that
$-\infty\leq c_{n}\leq\overline{c}$. We claim that $\varphi(c_{n})\downarrow0$.
If the sequence $\{c_{n}\}$ converges to some $c>-\infty$, the equality
$\alpha\varphi(c_{n})=\varphi(c_{n}-\lambda_{n})$ implies $\alpha\varphi(c)=\varphi(c)$
and then $\varphi(c)=0$. In the case $c_{n}\downarrow-\infty$, it
is clear that $\varphi(c_{n})\downarrow0$. Let $\{T_{k}\}$ be a
sequence of pairwise disjoint, measurable sets with $\mu(T_{k})<\infty$
and $\mu(T\setminus\bigcup_{k=1}^{\infty}T_{k})=0$. Thus we can select
a sub-sequence $\{c_{n_{k}}\}$ such that $\sum_{k=1}^{\infty}\varphi(c_{n_{k}})\mu(T_{k})<\infty$.
Let us define $c=\sum_{k=1}^{\infty}c_{n_{k}}\chi_{T_{k}}$ and $u_{0}=\sum_{k=1}^{\infty}\lambda_{n_{k}}\chi_{T_{k}}$.
From \eqref{eq:cn_def} it follows that
\[
\alpha\varphi(u)\leq\varphi(u-u_{0}(t)),\qquad\text{for all }u\geq c(t).
\]
Proposition~\ref{prop:condition_ineq} implies that $\varphi(\cdot)$
and $u_{0}$ satisfy condition \eqref{eq:u0_condition}.
\end{proof}

For the deformed exponential function $\varphi(\cdot)$ given in Example~\ref{exa:u0_not_satisfied},
it follows that
\begin{align*}
\limsup_{u\rightarrow\infty}\frac{\varphi(u)}{\varphi(u-\lambda_{0})} & =\limsup_{u\rightarrow\infty}\frac{e^{(u+1)^{2}/2}}{e^{(u-\lambda_{0}+1)^{2}/2}}\\
 & =\limsup_{u\rightarrow\infty}e^{u\lambda_{0}-(\lambda_{0}^{2}-\lambda_{0})/2}=\infty,
\end{align*}
which shows that $\varphi(\cdot)$ cannot be used in the generalization
of R\'{e}nyi divergence.

A deformed exponential function $\varphi(\cdot)$ that satisfies \eqref{eq:phi_cond_equiv}
does not increase faster then $u\mapsto e^{\lambda u}$ for some $\lambda\geq1$.
Expression \eqref{eq:phi_cond_equiv} is equivalent to the existence
of constants $K\geq1$ and $c\in\mathbb{R}\cup\{-\infty\}$ such that
\[
\frac{\varphi(u)}{\varphi(u-\lambda_{0})}\leq K,\qquad\text{for all }u\geq c.
\]
Fixed any $v\geq0$ we take an integer $n\geq0$ such that $n\lambda_{0}\leq v<(n+1)\lambda_{0}$.
For $u\geq c$, we can write
\begin{align*}
\varphi(u+v) & \leq\varphi(u+(n+1)\lambda_{0})\\
 & \leq K^{n+1}\varphi(u)\\
 & \leq K^{v/\lambda_{0}+1}\varphi(u)\\
 & =K\varphi(u)e^{\lambda v},
\end{align*}
where $\lambda=\log(K)/\lambda_{0}$. Therefore, the function $\varphi(\cdot)$
cannot increase faster then $u\mapsto e^{\lambda u}$ .

\subsection{Purely atomic case}

In this subsection, we will assume that $\mu$ is the counting measure
on the set of natural numbers $T=\mathbb{N}$. Due to this assumption,
notation changes a little. Sequences and summations are considered
in the place of functions and integrals. Condition \eqref{eq:u0_condition}
is rewritten as follows. We assume that there exists a sequence $\{u_{0,i}\}\subset(0,\infty)$
such that
\begin{equation}
\sum_{i=1}^{\infty}\varphi(c_{i}+\lambda u_{0,i})<\infty,\qquad\text{for all }\lambda>0,\label{eq:u0_condition_seq}
\end{equation}
for each sequence $\{c_{i}\}\subset\mathbb{R}$ such that $\sum_{i=1}^{\infty}\varphi(c_{i})<\infty$.
Beyond these changes, proofs of results involving condition \eqref{eq:u0_condition_seq}
require distinct techniques. In this subsection, we shall find an
equivalent criterion for a deformed exponential function and a sequence
$\{u_{0,i}\}$ to satisfy condition \eqref{eq:u0_condition_seq}.
We will prove that, in the case $\mu$ is the counting measure, any
deformed exponential function $\varphi(\cdot)$ admits a sequence
$\{u_{0,i}\}$ for which condition \eqref{eq:u0_condition_seq} holds.

\begin{prop}
\label{prop:condition_ineq_seq} A deformed exponential $\varphi\colon\mathbb{R}\rightarrow[0,\infty)$
and a sequence $\{u_{0,i}\}$ satisfy condition \eqref{eq:u0_condition_seq}
if, and only if, for some constants $\alpha\in(0,1)$ and $\varepsilon>0$,
we can find a sequence $\{c_{i}\}\subseteq\mathbb{\mathbb{R}}\cup\{-\infty\}$
such that $\sum_{i=1}^{\infty}\varphi(c_{i})<\infty$ and
\begin{equation}
\alpha\varphi(u)\leq\varphi(u-u_{0,i}),\qquad\text{for all }u>c_{i}\mbox{ with }\varphi(u-u_{0,i})<\varepsilon.\label{eq:ineq_condition_lemma_2_seq}
\end{equation}
\end{prop}

To prove Proposition~\ref{prop:condition_ineq_seq}, we require a
preliminary lemma.

\begin{lem}
\label{lem:technical_seq} Suppose that we cannot find $\alpha\in(0,1)$,
$\varepsilon>0$ and a sequence $\{c_{i}\}\subset\mathbb{R}\cup\{-\infty\}$
such that $\sum_{i=1}^{\infty}\varphi(c_{i})<\infty$ and
\begin{equation}
\alpha\varphi(u)\leq\varphi(u-u_{0,i}),\qquad\text{for all }u>c_{i}\mbox{ with }\varphi(u-u_{0,i})<\varepsilon.\label{eq:ineq_condition_lemma_seq}
\end{equation}
Then there exist sequences $\{\{c_{n,i}\}\}$ and $\{A_{n}\}$ of
finite-valued real numbers, and pairwise disjoint sets in $\mathbb{N}$,
respectively, such that
\begin{equation}
\frac{1}{2}\leq\sum_{i\in A_{n}}\varphi(c_{n,i})\quad\text{and}\quad\sum_{i\in A_{n}}\varphi(c_{n,i}-u_{0,i})\leq2^{-n},\label{eq:c_n_lambda_n_seq}
\end{equation}
for each $n\geq1$.
\end{lem}

\begin{proof}
For each $m\geq1$, we define the sequence $\{f_{m,i}\}\subset\mathbb{R}\cup\{-\infty\}$
by
\[
f_{m,i}=\sup\{u\in\mathbb{R}:2^{-m}\varphi(u)>\varphi(u-u_{0,i})\text{ and }\varphi(u-u_{0,i})\leq2^{-m-1}\},
\]
where we use the convention $\sup\emptyset=-\infty$. Since \eqref{eq:ineq_condition_lemma_seq}
is not satisfied, we have that $\sum_{i=1}^{\infty}\varphi(f_{m,i})=\infty$
for each $m\geq1$. We will consider the following cases.

\textit{Case}~1. There exists a strictly increasing sequence $\{m_{n}\}\subseteq\mathbb{N}$
for which the set $B_{n}=\{i:\varphi(f_{m_{n},i}-u_{0,i})=2^{-m_{n}-1}\}$
has an infinite number of elements. Then we can select a strictly
increasing sequence $\{i_{n}\}\subseteq\mathbb{N}$ such that
\[
2^{-m_{n}}\varphi(f_{m_{n},i_{n}})\geq\varphi(f_{m_{n},i_{n}}-u_{0,i_{n}})=2^{-m_{n}-1},
\]
which implies $\varphi(f_{m_{n},i_{n}})\geq1/2$. Expression \eqref{eq:c_n_lambda_n_seq}
follows with $c_{n,i}=f_{m_{n},i}$ and $A_{n}=\{i_{n}\}$.

\textit{Case}~2. There exists a strictly increasing sequence $\{m_{n}\}\subseteq\mathbb{N}$
for which the set $B_{n}$, as defined above, has a finite number
of elements. Let us denote $C_{n}=\mathbb{N}\setminus B_{n}=\{i:\varphi(f_{m_{n},i}-u_{0,i})<2^{-m_{n}-1}\}$.
By the continuity of $\varphi(\cdot)$, we have that $2^{-m_{n}}\varphi(f_{m_{n},i})=\varphi(f_{m_{n},i}-u_{0,i})$
for all $i\in C_{n}$. Because $\varphi(f_{m_{n},i})\leq1/2$ for
each $i\in C_{n}$, and $\sum_{i=1}^{\infty}\varphi(f_{m_{n},i})=\infty$
for all $n\geq1$, we can find a strictly increasing sequence $\{k_{n}\}\subset\mathbb{N}$
for which the set $A_{n}=C_{n}\cap\{k_{n-1},\dots,k_{n}-1\}$ satisfies
\[
\frac{1}{2}\leq\sum_{i\in A_{n}}\varphi(f_{m_{n},i})\leq1.
\]
The second inequality above in conjunction with $2^{-m_{n}}\varphi(f_{n,i})=\varphi(f_{n,i}-u_{0,i})$
implies
\[
\sum_{i\in A_{n}}\varphi(f_{m_{n},i}-u_{0,i})\leq2^{-m_{n}}.
\]
Thus expression \eqref{eq:c_n_lambda_n_seq} follows with $c_{n,i}=f_{m_{n},i}$.
\end{proof}

\begin{proof}[Proof of Proposition~\ref{prop:condition_ineq_seq}]
 To show that condition \eqref{eq:u0_condition_seq} implies inequality
\eqref{eq:ineq_condition_lemma_2_seq}, one can proceed as in the
proof of Proposition~\ref{prop:condition_ineq}, using Lemma~\ref{lem:technical_seq}
in the place of Lemma~\ref{lem:technical}.

Suppose that inequality \eqref{eq:ineq_condition_lemma_2_seq} is
satisfied. Let $\{\widetilde{c}_{i}\}$ be any sequence of real numbers
such that $\sum_{i=1}^{\infty}\varphi(\widetilde{c}_{i})<\infty$.
Denote $A=\{i:\widetilde{c}_{i}+u_{0,i}\geq c_{i}\}$ and $B=\{i\in A:\varphi(\widetilde{c}_{i})\leq\varepsilon\}$.
We use inequality \eqref{eq:ineq_condition_lemma_2} to write
\begin{align}
\alpha\sum_{i=1}^{\infty}\varphi(\widetilde{c}_{i}+u_{0,i}) & \leq\alpha\sum_{i\in A\cap B}\varphi(\widetilde{c}_{i}+u_{0,i})+\alpha\sum_{i\in A\setminus B}\varphi(\widetilde{c}_{i}+u_{0,i})+\alpha\sum_{i\in T\setminus A}\varphi(c_{i})\nonumber \\
 & \leq\sum_{i\in A}\varphi(\widetilde{c}_{i})+\alpha\sum_{i\in A\setminus B}\varphi(\widetilde{c}_{i}+u_{0,i})+\alpha\sum_{i\in T\setminus A}\varphi(c_{i})<\infty.\label{eq:condition_ineq_seq_step}
\end{align}
To conclude that the second summation in \eqref{eq:condition_ineq_seq_step}
is finite, we observed that the set $T\setminus B$ is finite. In
consequence, it follows that $\sum_{i=1}^{\infty}\varphi(\widetilde{c}_{i}+nu_{0,i})<\infty$
for all $n\geq1$; and then $\sum_{i=1}^{\infty}\varphi(\widetilde{c}_{i}+\lambda u_{0,i})<\infty$
for all $\lambda>0$.
\end{proof}

The result stated below shows that any deformed exponential function
$\varphi(\cdot)$ can be used in the generalization of R\'{e}nyi divergence,
in the case $\mu$ is the counting measure.

\begin{prop}
Let $\varphi\colon\mathbb{R}\rightarrow[0,\infty)$ be a deformed
exponential. Then we can find a sequence $\{u_{0,i}\}$ for which
condition \eqref{eq:u0_condition_seq} holds.
\end{prop}

\begin{proof}
Let $\{\lambda_{n}\}\subset(0,\infty)$ be any decreasing sequence
converging to $0$. Fix any $\alpha\in(0,1)$ and $\eta\in\mathbb{R}$
such that $\alpha\varphi(\eta)<\varphi(\eta-\lambda_{1})$. Denoting
$\varepsilon=\varphi(\eta-\lambda_{1})$, we define
\[
\widetilde{c}_{n}=\sup\{u\in\mathbb{R}:\alpha\varphi(u)>\varphi(u-\lambda_{n})\text{ and }\varphi(u-\lambda_{n})\leq\varepsilon\},\qquad\text{for each }n\geq1,
\]
where we adopt the convention $\sup\emptyset=-\infty$. Clearly, the
sequence $\{\widetilde{c}_{n}\}\subset[-\infty,\eta)$ is decreasing.
We claim that $\varphi(\widetilde{c}_{n})\downarrow0$. If the sequence
$\{\widetilde{c}_{n}\}$ converges to some $c>-\infty$, inequality
$\alpha\varphi(\widetilde{c}_{n})\geq\varphi(\widetilde{c}_{n}-\lambda_{n})$
implies $\alpha\varphi(c)\geq\varphi(c)$ and then $\varphi(c)=0$.
In the case $\widetilde{c}_{n}\downarrow-\infty$, it is clear that
$\varphi(\widetilde{c}_{n})\downarrow0$. Thus we can select a sub-sequence
$c_{i}=\widetilde{c}_{n_{i}}$ such that $\sum_{i=1}^{\infty}\varphi(c_{i})<\infty$
and
\[
\alpha\varphi(u)\leq\varphi(u-u_{0,i}),\qquad\text{for all }u>c_{i}\mbox{ with }\varphi(u)<\varepsilon,
\]
where $u_{0,i}=\lambda_{i}$. From Proposition~\ref{prop:condition_ineq_seq},
it follows that $\{u_{0,i}\}$ satisfies condition \eqref{eq:u0_condition_seq}.
\end{proof}

Such general models provide more robust methods to devise different distributions and improve the capability of inference of which the distribution better fits the available data. For example, in \cite{Rui:2019}, the authors employ a $\varphi$-divergence to the problem of image segmentation achieving better results than classical image processing methods. In their case, the selected $\varphi$ function complies the existence conditions discussed in this work. To fail meeting such existence conditions, in the problem of image classification (segmentation can be one step in the process) would lead to some classes of images (it would depend on the probability distribution of the images)  being erroneously assumed as different ones since the divergence would not include all the statistical characteristics of the image data.

\section{Conclusions}

This paper provided the existence conditions of a generalized R\'{e}nyi divergence so a deformed exponential function can be used to model the statistical distribution. Such conditions admit the design of a robust model by assuming any deformed exponential which provides the use of purely atomic measure. For the non-atomic case not all deformed exponentials can be used to generalize the R\'{e}nyi divergence although there are a fair amount of functions that comply with the existence conditions and therefore be applied to problems based on statistical divergence optimization. The results presented in this paper allow to consider discrete distributions (such as the one we can find in digital applications) to devise the differentiation between two probability distributions, which brings a greater number of possibilities of applications in several areas such as signal and image processing and possible extensions to quantum cases.

\section*{Acknowledgments}

The authors would like to thank CNPq (Procs.\ 408609/2016-8 and 309472/2017-2) and Coordena\c{c}\~{a}o de
Aperfei\c{c}oamento de Pessoal de N\'{\i}vel Superior - Brasil (CAPES) - Finance Code 001 for partial funding of this research.

\section*{References}

\bibliographystyle{elsarticle-num}
\bibliography{refs_conditions}

\begin{thebibliography}{10}
\expandafter\ifx\csname url\endcsname\relax
  \def\url#1{\texttt{#1}}\fi
\expandafter\ifx\csname urlprefix\endcsname\relax\def\urlprefix{URL }\fi
\expandafter\ifx\csname href\endcsname\relax
  \def\href#1#2{#2} \def\path#1{#1}\fi

\bibitem{Cover:2006}
T.~M. Cover, J.~A. Thomas, {Elements of Information Theory}, 2nd Edition,
  Wiley-Interscience [John Wiley \& Sons], Hoboken, NJ, 2006.

\bibitem{Shannon_1948}
C.~E. Shannon, {A mathematical theory of communication}, Bell System Tech. J.
  27 (1948) 379--423, 623--656.

\bibitem{Boltzmann_1964}
L.~Boltzmann, Lectures on gas theory, Translated by Stephen G. Brush,
  University of California Press, Berkeley-Los Angeles, Calif., 1964.

\bibitem{Tsallis:1988}
C.~Tsallis, {Possible generalization of {B}oltzmann-{G}ibbs statistics}, J.
  Statist. Phys. 52~(1-2) (1988) 479--487.
\newblock \href {http://dx.doi.org/10.1007/BF01016429}
  {\path{doi:10.1007/BF01016429}}.

\bibitem{Tsallis:1994}
C.~Tsallis, {What are the numbers that experiments provide?}, Quimica Nova
  17~(6) (1994) 468--471.

\bibitem{Suyari_2004}
H.~Suyari, {Generalization of Shannon-Khinchin axioms to nonextensive systems
  and the uniqueness theorem for the nonextensive entropy}, IEEE Transactions
  on Information Theory 50~(8) (2004) 1783--1787.
\newblock \href {http://dx.doi.org/10.1109/TIT.2004.831749}
  {\path{doi:10.1109/TIT.2004.831749}}.

\bibitem{Furuichi_2005}
S.~Furuichi, {On uniqueness Theorems for Tsallis entropy and Tsallis relative
  entropy}, IEEE Transactions on Information Theory 51~(10) (2005) 3638--3645.
\newblock \href {http://dx.doi.org/10.1109/TIT.2005.855606}
  {\path{doi:10.1109/TIT.2005.855606}}.

\bibitem{Tsallis_2001}
C.~Tsallis, {Nonextensive statistical mechanics and its applications}, Vol. 560
  of Lecture Notes in Physics, Springer-Verlag, Berlin, 2001.
\newblock \href {http://dx.doi.org/10.1007/3-540-40919-X}
  {\path{doi:10.1007/3-540-40919-X}}.

\bibitem{Renyi:1961}
A.~R{\'e}nyi, {On measures of entropy and information}, in: Proc. 4th
  {B}erkeley {S}ympos. {M}ath. {S}tatist. and {P}rob., {V}ol. {I}, Univ.
  California Press, Berkeley, Calif., 1961, pp. 547--561.

\bibitem{Principe:10}
J.~C. Pr{\'\i}ncipe, {Information Theoretic Learning: Renyi's Entropy and
  Kernel Perspectives}, Information Science and Statistics, Springer, 2010.

\bibitem{Zhang:2004}
J.~Zhang, {Divergence Function, Duality, and Convex Analysis}, Neural Comput.
  16~(1) (2004) 159--195.
\newblock \href {http://dx.doi.org/10.1162/08997660460734047}
  {\path{doi:10.1162/08997660460734047}}.

\bibitem{kullback1951}
S.~Kullback, R.~A. Leibler, On information and sufficiency, Ann. Math. Statist.
  22~(1) (1951) 79--86.
\newblock \href {http://dx.doi.org/10.1214/aoms/1177729694}
  {\path{doi:10.1214/aoms/1177729694}}.

\bibitem{Csiszar:2004}
I.~Csisz{\'a}r, P.~C. Shields, {Information theory and statistics: {A}
  tutorial}, Communications and Information Theory 1~(4) (2004) 417--528.

\bibitem{Tsallis_1998}
L.~Borland, A.~R. Plastino, C.~Tsallis, {Information gain within nonextensive
  thermostatistics}, Journal of Mathematical Physics 39~(12) (1998) 6490--6501.
\newblock \href {http://dx.doi.org/http://dx.doi.org/10.1063/1.532660}
  {\path{doi:http://dx.doi.org/10.1063/1.532660}}.

\bibitem{Furuichi:2004}
S.~Furuichi, K.~Yanagi, K.~Kuriyama, {Fundamental properties of {T}sallis
  relative entropy}, J. Math. Phys. 45~(12) (2004) 4868--4877.
\newblock \href {http://dx.doi.org/10.1063/1.1805729}
  {\path{doi:10.1063/1.1805729}}.

\bibitem{Hastie_2009}
T.~Hastie, R.~Tibshirani, J.~Friedman, {The Elements of Statistical Learning},
  2nd Edition, Springer Series in Statistics, Springer, New York, 2009, data
  mining, inference, and prediction.
\newblock \href {http://dx.doi.org/10.1007/978-0-387-84858-7}
  {\path{doi:10.1007/978-0-387-84858-7}}.

\bibitem{Principe_2010}
J.~C. Principe, {Information Theoretic Learning}, Information Science and
  Statistics, Springer, New York, 2010, renyi's entropy and kernel
  perspectives.
\newblock \href {http://dx.doi.org/10.1007/978-1-4419-1570-2}
  {\path{doi:10.1007/978-1-4419-1570-2}}.

\bibitem{Konishi_2008}
S.~Konishi, G.~Kitagawa, {Information Criteria and Statistical Modeling},
  Springer Series in Statistics, Springer, New York, 2008.
\newblock \href {http://dx.doi.org/10.1007/978-0-387-71887-3}
  {\path{doi:10.1007/978-0-387-71887-3}}.

\bibitem{Yamano:2009}
T.~Yamano, \href{https://doi.org/10.1063/1.3116115}{{A generalization of the
  Kullback-Leibler divergence and its properties}}, Journal of Mathematical
  Physics 50~(4) (2009) 043302.
\newblock \href {http://dx.doi.org/10.1063/1.3116115}
  {\path{doi:10.1063/1.3116115}}.
\newline\urlprefix\url{https://doi.org/10.1063/1.3116115}

\bibitem{Lionel:2013}
F.~{Pascal}, L.~{Bombrun}, J.~{Tourneret}, Y.~{Berthoumieu}, {Parameter
  Estimation For Multivariate Generalized Gaussian Distributions}, IEEE
  Transactions on Signal Processing 61~(23) (2013) 5960--5971.

\bibitem{Nielsen:2020}
F.~Nielsen, {On a Generalization of the Jensen-Shannon Divergence and the
  Jensen-Shannon Centroid}, Entropy 22~(2) (2020) 221.
\newblock \href {http://dx.doi.org/10.3390/e22020221}
  {\path{doi:10.3390/e22020221}}.

\bibitem{nock2}
R.~Nock, B.~Magdalou, E.~Briys, F.~Nielsen, {Mining Matrix Data with Bregman
  Matrix Divergences for Portfolio Selection}, Springer Berlin Heidelberg,
  Berlin, Heidelberg, 2013, pp. 373--402.

\bibitem{trivellato}
B.~Trivellato, {Deformed Exponentials and Applications to Finance}, Entropy
  15~(9) (2013) 3471--3489.

\bibitem{trivellato2}
E.~Moretto, S.~Pasquali, B.~Trivellato, {Option Pricing under Deformed Gaussian
  Distributions}, Physica A: Statistical Mechanics and its Applications 446
  (2016) 246--263.

\bibitem{CP}
A.~F.~P. Rodrigues, I.~M. Guerreiro, C.~C. Cavalcante, {Deformed exponentials
  and portfolio selection}, International Journal of Modern Physics C 29~(3).
\newblock \href {http://dx.doi.org/10.1142/S0129183118500298}
  {\path{doi:10.1142/S0129183118500298}}.

\bibitem{Flavia:2019}
A.~F.~P. Rodrigues, C.~C. Cavalcante, V.~L. Cris{\'o}stomo,
  \href{http://www.sciencedirect.com/science/article/pii/S0378437119312646}{{A
  projection pricing model for non-Gaussian financial returns}}, Physica A:
  Statistical Mechanics and its Applications 534 (2019) 122181.
\newblock \href {http://dx.doi.org/https://doi.org/10.1016/j.physa.2019.122181}
  {\path{doi:https://doi.org/10.1016/j.physa.2019.122181}}.
\newline\urlprefix\url{http://www.sciencedirect.com/science/article/pii/S0378437119312646}

\bibitem{Alvarez:97}
J.~Sala-Alvarez, G.~V{\'a}zquez-Grau, {Statistical Reference Criteria for
  Adaptive Signal Processing in Digital Communications}, IEEE Transactions on
  Signal Processing Vol. 45~(No. 1) (1997) 14--31.

\bibitem{Santamaria:02}
I.~{Santamar{\'\i}a}, C.~{Pantale{\'o}n}, L.~{Vielva}, J.~C. {Principe}, {Fast
  algorithm for adaptive blind equalization using order-$\alpha$ Renyi's
  entropy}, in: 2002 IEEE International Conference on Acoustics, Speech, and
  Signal Processing (ICASSP), Vol.~3, Orlando, FL, USA, 2002, pp.
  III--2657--III--2660.

\bibitem{Cavalcante:02}
C.~C. Cavalcante, F.~R.~P. Cavalcanti, J.~C.~M. Mota, {Adaptive Blind Multiuser
  Separation Criterion Based on Log-Likelihood Maximisation}, {IEE Electronics
  Letters} 38~(20) (2002) 1231--1233.

\bibitem{Cavalcante:05_SP_Elsevier}
C.~C. Cavalcante, J.~M.~T. Romano, {Multi-user pdf Estimation Based Criteria
  for Adaptive Blind Separation of Discrete Sources}, Signal Processing 85~(5)
  (2005) 1059--1072.

\bibitem{Yannick:2013}
A.~M. {Atto}, E.~{Trouve}, Y.~{Berthoumieu}, G.~{Mercier}, {Multidate
  Divergence Matrices for the Analysis of SAR Image Time Series}, IEEE
  Transactions on Geoscience and Remote Sensing 51~(4) (2013) 1922--1938.

\bibitem{Fatima:2016}
R.~H. {Nobre}, F.~A.~A. {Rodrigues}, R.~C.~P. {Marques}, J.~S. {Nobre}, J.~F.
  S.~R. {Neto}, F.~N.~S. {Medeiros}, {SAR Image Segmentation With R{\'e}nyi's
  Entropy}, IEEE Signal Processing Letters 23~(11) (2016) 1551--1555.

\bibitem{Rui:2019}
J.~B. Barreto, R.~F. Vigelis, {Clusteriza{\c{c}}{\~a}o Baseada na
  $\varphi$-Diverg{\^e}ncia Aplicada {\`a} Segmenta{\c{c}}{\~a}o de Imagens},
  in: Proc. of XXXVII Simp{\'o}sio Brasileiro de Telecomunica{\c{c}}{\~o}es e
  Processamento De Sinais (SBrT2019), Petr{\'o}polis, RJ - Brazil, 2019.

\bibitem{Nielsen:2011}
B.~C. {Vemuri}, M.~{Liu}, S.~{Amari}, F.~{Nielsen}, {Total Bregman Divergence
  and Its Applications to DTI Analysis}, IEEE Transactions on Medical Imaging
  30~(2) (2011) 475--483.

\bibitem{Abe_2003}
S.~Abe, \href{https://link.aps.org/doi/10.1103/PhysRevA.68.032302}{{Nonadditive
  generalization of the quantum Kullback-Leibler divergence for measuring the
  degree of purification}}, Phys. Rev. A 68 (2003) 032302.
\newblock \href {http://dx.doi.org/10.1103/PhysRevA.68.032302}
  {\path{doi:10.1103/PhysRevA.68.032302}}.
\newline\urlprefix\url{https://link.aps.org/doi/10.1103/PhysRevA.68.032302}

\bibitem{Luiza:2019}
L.~H.~F. Andrade,
  \href{http://aimsciences.org//article/id/4dccc59b-b0b1-4882-9fe6-8fdb418a24fd}{A
  generalized quantum relative entropy} (2019).
\newblock \href {http://dx.doi.org/10.3934/amc.2020063}
  {\path{doi:10.3934/amc.2020063}}.
\newline\urlprefix\url{http://aimsciences.org//article/id/4dccc59b-b0b1-4882-9fe6-8fdb418a24fd}

\bibitem{Amari1982}
S.-I. Amari, {Differential Geometry of Curved Exponential Families-Curvatures
  and Information Loss}, The Annals of Statistics 10~(2) (1982) 357--385.

\bibitem{Amari1985}
S.-i. Amari, {Differential-geometrical methods in statistics}, Vol.~28 of
  Lecture Notes in Statistics, Springer-Verlag, New York, 1985.
\newblock \href {http://dx.doi.org/10.1007/978-1-4612-5056-2}
  {\path{doi:10.1007/978-1-4612-5056-2}}.

\bibitem{Amari:2000}
S.-i. Amari, H.~Nagaoka, Methods of information geometry, Vol. 191 of
  Translations of Mathematical Monographs, American Mathematical Society,
  Providence, RI; Oxford University Press, Oxford, 2000, translated from the
  1993 Japanese original by Daishi Harada.

\bibitem{Bregman1967}
L.~Bregman, The relaxation method of finding the common point of convex sets
  and its application to the solution of problems in convex programming,
  \{USSR\} Computational Mathematics and Mathematical Physics 7~(3) (1967) 200
  -- 217.
\newblock \href
  {http://dx.doi.org/http://dx.doi.org/10.1016/0041-5553(67)90040-7}
  {\path{doi:http://dx.doi.org/10.1016/0041-5553(67)90040-7}}.

\bibitem{Inequalities_1952}
G.~H. Hardy, J.~E. Littlewood, G.~P{\'o}lya, Inequalities, Cambridge,
  University Press, 1952, 2d ed.

\bibitem{Zhang_2015}
J.~Zhang, {On Monotone Embedding in Information Geometry}, Entropy 17~(7)
  (2015) 4485--4499.
\newblock \href {http://dx.doi.org/10.3390/e17074485}
  {\path{doi:10.3390/e17074485}}.

\bibitem{Naudts_Zhang_2018}
J.~Naudts, J.~Zhang, {Rho--tau embedding and gauge freedom in information
  geometry}, Information Geometry 1~(1) (2018) 79--115.
\newblock \href {http://dx.doi.org/10.1007/s41884-018-0004-6}
  {\path{doi:10.1007/s41884-018-0004-6}}.

\bibitem{Chhabra_2016}
K.~Jain, P.~Chhabra,
  \href{http://www.inderscienceonline.com/doi/abs/10.1504/IJICOT.2016.076964}{New
  generalised divergence measure for increasing functions}, International
  Journal of Information and Coding Theory 3~(3) (2016) 197--216.
\newblock \href
  {http://arxiv.org/abs/http://www.inderscienceonline.com/doi/pdf/10.1504/IJICOT.2016.076964}
  {\path{arXiv:http://www.inderscienceonline.com/doi/pdf/10.1504/IJICOT.2016.076964}},
  \href {http://dx.doi.org/10.1504/IJICOT.2016.076964}
  {\path{doi:10.1504/IJICOT.2016.076964}}.
\newline\urlprefix\url{http://www.inderscienceonline.com/doi/abs/10.1504/IJICOT.2016.076964}

\bibitem{Naudts_2002}
J.~Naudts, {Deformed exponentials and logarithms in generalized
  thermostatistics}, Physica A: Statistical Mechanics and its Applications 316.
\newblock \href {http://dx.doi.org/10.1016/s0378-4371(02)01018-x}
  {\path{doi:10.1016/s0378-4371(02)01018-x}}.

\bibitem{naudts5}
J.~Naudts, {Deformed exponentials and logarithms in generalized
  thermostatistics}, Physica A: Statistical Mechanics and its Applications
  316~(1-4) (2002) 323--334.

\bibitem{Vigelis:2013}
R.~F. Vigelis, C.~C. Cavalcante, {On {$\varphi$}-families of probability
  distributions}, Journal of Theoretical Probability 26~(3) (2013) 870--884.
\newblock \href {http://dx.doi.org/10.1007/s10959-011-0400-5}
  {\path{doi:10.1007/s10959-011-0400-5}}.

\bibitem{Korbel_2019}
J.~Korbel, R.~Hanel, S.~Thurner, {Information Geometric Duality of
  $\phi$-Deformed Exponential Families}, Entropy 21~(2).
\newblock \href {http://dx.doi.org/10.3390/e21020112}
  {\path{doi:10.3390/e21020112}}.

\bibitem{Erven:2014}
T.~van Erven, P.~Harremo{\"e}s,
  \href{http://dx.doi.org/10.1109/TIT.2014.2320500}{{R\'enyi divergence and
  {K}ullback-{L}eibler divergence}}, IEEE Trans. Inform. Theory 60~(7) (2014)
  3797--3820.
\newblock \href {http://dx.doi.org/10.1109/TIT.2014.2320500}
  {\path{doi:10.1109/TIT.2014.2320500}}.
\newline\urlprefix\url{http://dx.doi.org/10.1109/TIT.2014.2320500}

\bibitem{Souza:2016}
D.~C. de~Souza, R.~F. Vigelis, C.~C. Cavalcante, {Geometry induced by a
  generalization of {R}\'enyi divergence}, Entropy 18~(11) (2016) Paper No.
  407, 16.
\newblock \href {http://dx.doi.org/10.3390/e18110407}
  {\path{doi:10.3390/e18110407}}.

\bibitem{Kullback:1951}
S.~Kullback, R.~A. Leibler, {On information and sufficiency}, Ann. Math.
  Statistics 22 (1951) 79--86.

\bibitem{Kaniadakis:2002}
G.~Kaniadakis, {Statistical mechanics in the context of special relativity},
  Phys. Rev. E (3) 66~(5) (2002) 056125, 17.
\newblock \href {http://dx.doi.org/10.1103/PhysRevE.66.056125}
  {\path{doi:10.1103/PhysRevE.66.056125}}.

\bibitem{Musielak:1983}
J.~Musielak, Orlicz spaces and modular spaces, Vol. 1034 of Lecture Notes in
  Mathematics, Springer-Verlag, Berlin, 1983.

\end{thebibliography}

\end{document}